\documentclass{llncs}

\usepackage{bm}       
\usepackage{amssymb}  

\title{Synchronous Games, Simulations and $\lambda$-calculus}
\author{Pierre Hyvernat\inst{1,2}}
\institute{Institut math\'ematique de Luminy, Marseille, France\\
           \and
           Chalmers Institute of Technology, G\"oteborg, Sweden\\
           \email{hyvernat@iml.univ-mrs.fr}}

\makeatletter
\let\ea=\expandafter

\def\tvi{\vrule height12pt depth5pt width0pt}

\newbox\hyp \newbox\concl \newdimen\ruleWidth

\def\rule@nd{\rule@end}

\def\rule@ux#1&{
 \def\tmp{#1}
  \ifx\tmp\rule@nd
   \hskip-2em$
    \else\tmp\hskip2em\ea\rule@ux
     \fi}

\def\infer#1#2#3{\relax
 \setbox\hyp=\hbox{$\rule@ux#1&\rule@end&} \setbox\concl=\hbox{$#2$}
  \ifdim \wd\concl<\wd\hyp
   \ruleWidth=\wd\hyp
    \else\ruleWidth=\wd\concl
     \fi
     \advance\ruleWidth by 1cm
     $\vcenter{
       \vbox{\hbox to\ruleWidth{\hss\tvi \unhbox\hyp \hss}}
        \hrule height.7pt depth0pt width\ruleWidth
         \vbox{\hbox to\ruleWidth{\hss\tvi \unhbox\concl \hss}}}$
          \kern.1cm \hbox{#3}}
\makeatother

\newcommand\ie{\mbox{\textit{i.e.}}~}
\newcommand\etc{\mbox{\textrm{etc.}}}

\newcommand{\Int}{\mathbf{Int}}          
\newcommand{\Mulf}{\mathcal{M}_{\!f}}    
\newcommand\D{\mathop{{}\mathrm{D}}}
\newcommand\AC{\mathsf{AC}}
\newcommand\PI{\mathsf{\Pi}^1}             
\newcommand\Bool{\mathbf{B}}
\newcommand\T{\mathsf{t}}
\newcommand\F{\mathsf{f}}
\newcommand\List{\mathsf{List}}
\newcommand\Skip{\mathsf{skip}}
\newcommand\Magic{\mathsf{magic}}
\newcommand\Abort{\mathsf{abort}}
\newcommand\LL{\mbox{\L}}
\newcommand\Perm{\mathfrak{S}}

\newcommand{\Tensor}{\mathbin{\otimes}}
\newcommand{\Bottom}{\bot}
\newcommand{\One}{\mathbf{1}}
\newcommand{\LinearArrow}{\mathbin{\multimap}}
\newcommand{\Plus}{\mathbin{\oplus}}
\DeclareSymbolFont{fontexp}{OT1}{cmr}{bx}{n}
\DeclareMathSymbol{!}{\mathalpha}{fontexp}{`!}

\newcommand\sitem[1]{\item[\small{\textit{(#1)}}]}
\newcommand\sit[1]{{\small\textit{(#1)}}}

\newcommand\be{\[\begin{array}[t]{llllllll}} \newcommand\ee{\end{array}\]}
\newcommand\bce{\[\begin{array}[t]{c}} \newcommand\ece{\end{array}\]}

\newcommand\step[2]{\smallbreak\hskip.25cm$#1$\ifx\empty#2\else\hskip.25cm$\{$ {\small #2} $\}$\fi\smallbreak}

\begin{document}

\maketitle

\begin{abstract}
  We refine a model for linear logic based on two well-known ingredients:
  games and simulations. We have already shown that usual simulation relations
  form a sound notion of morphism between games; and that we can interpret all
  linear logic in this way. One particularly interesting point is that we
  interpret multiplicative connectives by synchronous operations on games.

  We refine this work by giving computational contents to our simulation
  relations. To achieve that, we need to restrict to intuitionistic linear
  logic. This allows to work in a constructive setting, thus keeping a
  computational content to the proofs.

  We then extend it by showing how to interpret some of the additional
  structure of the exponentials.

  \smallbreak
  To be more precise, we first give a denotational model for the typed
  $\lambda$-calculus; and then give a denotational model for the differential
  $\lambda$-calculus of Ehrhard and Regnier. Both this models are proved
  correct constructively.
\end{abstract}

\section*{Introduction} 

Transition systems and simulation relations are well known tools in computer
science. More recent is the use of games to give models for different
programming languages \cite{AJM,HO,Algol}, or as an interesting tool for the
study of other programming notions \cite{ghica}. We have devised in
\cite{ISLL} a denotational model of linear logic based on those two ideas.
Basically, a formula was interpreted by an alternating transition system
(called an \emph{interaction system}) and a proof was interpreted by a
\emph{safety property} for this interaction system.  Those concepts which were
primarily developed to model imperative programming and interfaces turned out
to be a rather interesting games model: a formula is interpreted by a game
(the interaction systems), and a proof by a ``non-loosing strategy'' (the
safety property).

Part of the interest is that the notion of safety property is very simple: it
is only a subset of the set of states.  However, in terms of games, the
associated strategy (whose existence is guaranteed by the condition satisfied
by the subset of states) is usually not computable.  We will show that it is
possible to overcome this problem by restricting to intuitionistic linear
logic.  More precisely, we will model typed $\lambda$-calculus (seen as a
subsystem of intuitionistic linear logic) within a constructive setting. The
model for full intuitionistic linear logic (ILL) can easily be derived the
present work and the additive connectives defined in \cite{ISLL}.

\smallbreak
The structure of safety properties is in fact richer than the structure of
$\lambda$-terms. In particular, safety properties are closed under unions.
Since there is no sound notion of ``logical sum'' of proofs, this doesn't
reflect a logical property.  However, it is important in programming since it
can be used to interpret non-determinism. The differential $\lambda$-calculus
of Ehrhard and Regnier (\cite{diffLamb}) is an extension to the
\mbox{$\lambda$-calculus,} which has a notion of non deterministic sum. We
show how to interpret this additional structure.

\section{Interaction Systems} 

\subsection{The Category of Interaction Systems} 

We briefly recall the important definitions. For more motivations, we refer to
\cite{giovanni} and \cite{ISLL}.

\begin{definition}
  Let $S$ be a set (of \emph{states}); an \emph{interaction system} on $S$ is
  given by the following data:
  \begin{itemize}
    \item for each $s\in S$, a set $A(s)$ of \emph{possible actions;}
    \item for each $a\in A(s)$, a set $D(s,a)$ of \emph{possible reactions} to
      $a$;
    \item for each $d\in D(s,a)$, a \emph{new state} $n(s,a,d)\in S$.
  \end{itemize}
  We usually write $s[a/d]$ instead of $n(s,a,d)$.
\end{definition}
Following standard practise within computer science, we distinguish the two
``characters'' by calling them the Angel (choosing actions, hence the $A$) and
the Demon (choosing reactions, hence the $D$).  Depending on the authors'
background, other names could be Player and Opponent, Eloise and Abelard,
Alice and Bob, Master and Slave, Client and Server, System and Environment,
\etc

\medbreak
One of the original goals for interaction systems (Hancock) was to represent
real-life programming interfaces. Here is for example the interface of a stack
of booleans:
\begin{itemize}
  \item $S = \List(\Bool)$;
  \item $A(\_) = \{\mathsf{Push}(b) \mid b\in \Bool\} \cup \{\mathsf{Pop}\}$;
  \item $\left\{\begin{array}[c]{l}
           D(\_,\mathsf{Push}(b)) = \{*\} \\
           D([],\mathsf{Pop}) = \{\mathsf{error}\}\\
           D(b:s,\mathsf{Pop}) = \{*\}
         \end{array}\right.$
  \item $\left\{\begin{array}[c]{l}
           n(s,\mathsf{Push}(b)) = b:s  \\
           n([],\mathsf{Pop},\mathsf{error}) = []\\
           n(b:s,\mathsf{Pop}) = s
         \end{array}\right.$
\end{itemize}
This gives in full details the specification of the stack interface. This is
more precise than classical interfaces which are usually given by a
collection of types: compare with this poor description of stacks:
\begin{itemize}
  \item $\mathsf{Pop} : \Bool$
  \item $\mathsf{Push} : \Bool \rightarrow ()$
\end{itemize}
which doesn't specify what the command actually \emph{do}; but only tells how
they can be used.

\medbreak
The notion of morphism between such interaction systems is an extension of the
usual notion of simulation relation:
\begin{definition}
  If $w_1$ and $w_2$ are two interaction systems on $S_1$ and $S_2$
  respectively; a relation $r\subseteq S_1 \times S_2$ is called a simulation if:
  \be
    (s_1,s_2)\in r \quad \Rightarrow \quad & \big(\forall a_1\in A_1(s_1)\big)\\
                                  & \big(\exists a_2\in A_2(s_2)\big)\\
                                  & \big(\forall d_2\in D_2(s_2,a_2)\big)\\
                                  & \big(\exists d_1\in D_1(s_1,a_1)\big)\\
                                  & \qquad \big(s_1[a_1/d_1],s_2[a_2/d_2]\big)\in r \ \mbox{.}
  \ee
\end{definition}
This definition is very similar to the usual definition of simulation relation
between labelled transition systems, but adds one layer of quantifiers to deal
with reactions. That $(s_1,s_2) \in r$ means that ``$s_2$ simulates $s_1$''.
By extension, if $a_2$ is a witness to the first existential quantifier, we
say that ``$a_2$ simulates $a_1$''.  Note that the empty relation is
\emph{always} a simulation. In practise, to prevent this degenerate case, we
would add a notion of initial state(s) and require that initial states are
related through the simulation.

\smallbreak
To continue on the previous example, programming a stack interface amounts to
implementing the stack commands using a lower level interface (arrays and
pointer for examples). If we interpret the quantifiers constructively, this
amounts to providing a (constructive) proof that a non-empty relation is a
simulation from this lower level interaction system to stacks. (See
\cite{giovanni} for a more detailed description of programming in terms of
interaction systems.)

\medbreak\noindent
Recall that the composition of two relations is given by:
\[
  (s_1,s_3) \in r_2 \cdot r_1 \Leftrightarrow (\exists s_2)\ (s_1,s_2)\in r_1 \ \mbox{and}\
  (s_2,s_3)\in r_2
\]
It should be obvious that the composition of two simulations is a simulation
and that the equality relation is a simulation from any $w$ to itself.  Thus,
we can put:
\begin{definition}
  We call $\Int$ the category of interaction systems with simulations.
\end{definition}
Note that everything has a computational content: the composition of two
simulations is just given by the composition of the two ``algorithms''
simulating $w_3$ by $w_2$ and $w_2$ by $w_1$; and that the algorithm for the
identity from $w$ to $w$ is simply the ``copycat'' strategy.

\subsection{Notation} 

Before diving in the structure of interaction systems, let's detail some of
the notation.
\begin{itemize}

  \item An element of the indexed cartesian product $\prod_{a\in A} D(a)$ is
    given by a function $f$ taking any $a\in A$ to an $f(a)$ in $D(a)$. When the
    set $D(a)$ doesn't depend on $a$, it amounts to a function $f:A \rightarrow D$.

  \item An element of the indexed disjoint sum $\sum_{a\in A} D(a)$ is given
    by a pair $(a,d)$ where $a\in A$ and $d\in D(a)$. When the set $D(a)$
    doesn't depend on $a$, this is simply the cartesian product $A \times D$.

  \item We write $\List(S)$ for the set of ``lists'' over set $S$. A list is
    simply a tuple $(s_1,s_2, \ldots s_n)$ of elements of $S$. The empty list
    is denoted $()$.

  \item The collection $\Mulf(S)$ of finite multisets over $S$ is the quotient
    of $\List(S)$ by permutations. We write $[s_1,\ldots s_n]$ for the
    equivalence class containing $(s_1,\ldots s_n)$.
    We write ``$+$'' for the sum of multisets. It simply corresponds to
    concatenation on lists.
\end{itemize}
Concerning the product and sum operators, it should be noted that they have a
computational content if one works in a constructive setting: an element of
$\prod_{a\in A} D(a)$ is an algorithm with input $a\in A$ and output
$f(a)\in D(a)$; and an element of $\sum_{a\in A}D(a)$ is simply a pair as
above. This is in fact the basis of dependent type theory frameworks like
Martin-L\"of's type theory or the calculus of construction.

\smallbreak
\noindent
\textbf{Remark:} even if it was an important motivation for this work, we do
not insist too much on the ``constructive mathematics'' part.  Readers
familiar with constructive frameworks should easily see that everything makes
computational sense; and classical readers can skip the comments about
computational content.

\subsection{Constructions} 

We now define the connectives of multiplicative exponential linear logic.
With those, making $\Int$ into a denotational model of intuitionistic
multiplicative exponential linear logic more or less amounts to showing that it
is symmetric monoidal closed, with a well behaved comonad.

\subsubsection{Constant.}

A very simple, yet important interaction system is ``$\Skip$'', the
interaction system without interaction. Following the linear logic
convention, we call it $\Bottom$:
\begin{definition}
  Define $\Bottom$ (or $\Skip$) to be the following interaction system on the
  Singleton set $\{*\}$:
  \be
    A_\Bottom(*)     &=& \{*\}\\
    D_\Bottom(*,*)   &=& \{*\}\\
    n_\Bottom(*,*,*) &=& \{*\} \ \hbox{.}
  \ee
  Depending on the context, this interaction system is also denoted by $\One$.
\end{definition}
Note that it is very different from the two following interaction systems (on
the same set of states) which respectively deadlock the Angel and the Demon:
\be
 A_a(*)       &=& \emptyset &\qquad& A_d(*)      &=& \{*\} \\
 D_a(*,\_)    &=& \_        &      & D_d(*,*)    &=& \emptyset \\
 n_a(*,\_,\_) &=& \_        &      & n_d(*,*,\_) &=& \_         & \hbox{.}
\ee
Those two systems play an important r\^ole in the general theory of
interaction systems (the first one is usually called $\Abort$, while the
second one is usually called $\Magic$) but they do not appear in the model
presented below.

\subsubsection{Synchronous Product.}

There is an obvious product construction reminiscent of the synchronous
product found in SCCS (synchronous calculus of communicating systems,
\cite{SCCS}):
\begin{definition}
  Suppose $w_1$ and $w_2$ are interaction systems on $S_1$ and $S_2$.
  Define the interaction system $w_1\Tensor w_2$ on $S_1 \times S_2$ as follows:
  \be
    A_{w_1\Tensor w_2}\big((s_1,s_2)\big)                     & \quad=\quad & A_1(s_1) \times A_2(s_2)\\
    D_{w_1\Tensor w_2}\big((s_1,s_2),(a_1,a_2)\big)           & \quad=\quad & D_1(s_1,a_1) \times D_2(s_2,a_2)\\
    n_{w_1\Tensor w_2}\big((s_1,s_2),(a_1,a_2),(d_1,d_2)\big) & \quad=\quad & \big(s_1[a_1/d_1],s_2[a_2/d_2]\big) \ \mbox{.}
  \ee
\end{definition}
This is the \emph{synchronous parallel composition} of $w_1$ and $w_2$: the
Angel and the Demon exchange pairs of actions/reactions.

For any sensible notion of morphism, $\Skip$ should be a neutral element for
this product. It is indeed the case, for the following reason: the components
of $w\Tensor\Skip$ and $w$ are isomorphic by dropping the second (trivial)
coordinate:
\be
  w\Tensor \One                & &                        &      & w\\
  \noalign{\smallbreak}
  S \times \{*\}                      & &                        &\qquad& S\\
  A\big((s,*)\big)             &=& A(s) \times \{*\}           &      & A(s)\\
  D\big((s,*),(a,*)\big)       &=& D(s,a) \times \{*\}         &      & D(s,a)\\
  n\big((s,*),(a,*),(d,*)\big) &=& \big(s[a/d],*\big)     &      & s[a/d]\\
\ee
This implies trivially that $\{( (s,*),s) \mid s\in S\}$ is an isomorphism.  For
similar reasons, this product is transitive and commutative.

\begin{lemma}
  ``$\_\Tensor\_$'' is a commutative tensor product in the category~$\Int$. Its
  action on morphisms is given by:
  \[
  \big((s_1,s'_1),(s_2,s'_2)\big) \in r\Tensor r'
 \Leftrightarrow 
  \left\{\begin{array}[c]{ll}
                             &(s_1,s_2)\in r\\
           \mbox{\small and} &(s'_1,s'_2)\in r'
         \end{array}\right.
  \]
\end{lemma}
Checking that $r\Tensor r'$ is indeed a simulation is easy.

\medbreak
Note that not every isomorphism (in the category $\Int$) is of this form: is
is quite simple to find isomorphic interaction systems with non-isomorphic
components.\footnote{In the finite case, one can duplicate a command $a$ into
$a_1$ and $a_2$ to obtain sets of commands of different cardinality.}

\subsubsection{Linear Arrow.}

The definition of the interaction system $w_1 \LinearArrow w_2$ is not as obvious as the
definition of the tensor ($\Tensor$):
\begin{definition}
  If $w_1$ and $w_2$ are interaction systems on $S_1$ and $S_2$, define the
  interaction system $w_1 \LinearArrow w_2$ on $S_1 \times S_2$ as follows:
  \be
  A\big((s_1,s_2)\big) = \displaystyle
                           \sum_{f\in A_1(s_1) \rightarrow A_2(S_2)}\ 
                           \prod_{a_1\in A_1(s_1)}
                           D_2\big(s_2,f(a_1)\big) \rightarrow D_1(s_1,a_1)\\
  D\big((s_1,s_2),(f,G)\big) = \displaystyle\sum_{a_1\in A_1(s_1)} D_2\big(s_2,f(a_1)\big)\\
  n\big((s_1,s_2),(f,G),(a_1,d_2)\big) = \big(s_1[a_1/G_{a_1}(d_2)]\,,\ s_2[f(a_1)/d_2]\big)
  \ \hbox{.}
  \ee
\end{definition}
It may seem difficult to get some intuition about this interaction system; but
it is \textit{a posteriori} quite natural: (see Proposition~\ref{prop:adjoint})
\begin{itemize}
  \item An action in state $(s_1,s_2)$ is given by:
    \begin{itemize}
      \sitem1 a function $f$ (the index for the element of the disjoint sum)
      translating actions from $s_1$ into actions from $s_2$;
      \sitem2 for any action $a_1$, a function $G_{a_1}$ translating reactions
        to $f(a_1)$ into reactions to $a_1$.
    \end{itemize}
  \item A reaction to such a ``translating mechanism'' is given by:
    \begin{itemize}
      \sitem1 an action $a_1$ in $A_1(s_1)$ (which we want to simulate);
      \sitem2 and a reaction $d_2$ in $D_2(s_2,f(a_1))$ (which we want to
        translate back).
    \end{itemize}
  \item Given such a reaction, we can simulate $a_1$ by $a_2\in A_2(s_2)$
    obtained by applying $f$ to $a_1$; and translate back $d_2$ into $d_1\in
    D_1(s_1,a_1)$ by applying $G_{a_1}$ to $d_2$. The next state is just the
    pair of states $s_1[a_1/d_1]$ and $s_2[a_2/d_2]$.
\end{itemize}

It thus looks like the interaction system $w_1 \LinearArrow w_2$ is related to
simulations from $w_1$ to $w_2$. It is indeed the case:
\begin{proposition}  \label{prop:adjoint}
  In $\Int$, ``$\_\Tensor\_$'' is left adjoint to ``$\_ \LinearArrow \_$''.
\end{proposition}

\begin{proof}
  The proof is not really difficult, but is quite painful to write (or read). Here is an
  attempt.

  \smallbreak
  \noindent
  Note that the following form of the axiom of choice is constructively
  valid:\footnote{This form of the axiom of choice is provable in Martin-L\"of's type theory or in
  the calculus of construction...}
  \be
  \AC :\quad
  \big(\forall a\in A\big)\big(\exists d\in D(a)\big) \varphi(a,d)
 \Leftrightarrow 
  \big(\exists f\in\prod_{a\in A} D(a)\big)\big(\forall a\in A\big)
  \varphi\big(a,f(a)\big)
  \ee
  When the domain $D(a)$ for the existential quantifier doesn't depend on
  $a\in A$, we can simplify it into:
  \be
  \AC :\quad
  \big(\forall a\in A\big)\big(\exists d\in D\big) \varphi(a,d)
 \Leftrightarrow 
  \big(\exists f\in A \rightarrow D\big)\big(\forall a\in A\big) \varphi\big(a,f(a)\big)
  \ee

  In the sequel, the part of the formula being manipulated will be written in
  bold.  That $r$ is a simulation from $w_1\Tensor w_2$ to $w_3$ takes the
  form\footnote{modulo associativity $(S_1 \times S_2) \times S_3 \simeq S_1 \times (S_2 \times S_3)\simeq
  S_1 \times S_2 \times S_3$...}
  \be
    (s_1,s_2,s_3)\in r & \Rightarrow & \big(\forall a_1\in A_1(s_1)\big)\bm{\big(\forall a_2\in A_2(s_2)\big)}\\
                       &  & \bm{\big(\exists a_3\in A_3(s_3)\big)}\\
                       &  & \big(\forall d_3\in D_3(s_3,\bm{a_3})\big)\\
                       &  & \big(\exists d_1\in D_1(s_1,a_1)\big)\big(\exists d_2\in D_2(s_2,a_2)\big)\\
                       &  & \quad \big(s_1[a_1/d_1],s_2[a_2/d_2],s_3[\bm{a_3}/d_3]\big) \in r \ \hbox{.}
  \ee
  Using $\AC$ on the $\forall a_2\exists a_3$, we obtain:
  \be
  (s_1,s_2,s_3)\in r & \Rightarrow & \big(\forall a_1\in A_1(s_1)\big)\\
                     &  & \big(\exists f\in A_2(s_2) \rightarrow A_3(s_3)\big)\\
                     &  & \big(\forall a_2\in A_2(s_2)\big)\bm{\big(\forall d_3\in D_3(s_3,f(a_2))\big)}\\
                     &  & \big(\exists d_1\in D_1(s_1,a_1)\big)\bm{\big(\exists d_2\in D_2(s_2,a_2)\big)}\\
                     &  & \quad \big(s_1[a_1/d_1],s_2[a_2/\bm{d_2}],s_3[f(a_2)/d_3]\big) \in r \ \hbox{.}
  \ee
  We can now apply $\AC$ on $\forall d_3\exists d_2$:
  \be
  (s_1,s_2,s_3)\in r & \Rightarrow & \big(\forall a_1\in A_1(s_1)\big)\\
                     &  & \big(\exists f\in A_2(s_2) \rightarrow A_3(s_3)\big)\\
                     &  & \bm{\big(\forall a_2\in A_2(s_2)\big)}\\
                     &  & \bm{\big(\exists g\in D_3(s_3,f(a_2)) \to D_2(s_2,a_2)\big)}\\
                     &  & \big(\forall d_3\in D_3(s_3,f(a_2))\big)\\
                     &  & \big(\exists d_1\in D_1(s_1,d_1)\big)\\
                     &  & \quad \big(s_1[a_1/d_1],s_2[a_2/\bm{g}(d_3)],s_3[f(a_2)/d_3]\big) \in r
  \ee
  and apply $\AC$ one more time on $\forall a_2\exists g$ to obtain:
  \be
  (s_1,s_2,s_3)\in r & \Rightarrow & \big(\forall a_1\in A_1(s_1)\big)\\
                     &  & \big(\exists f\in A_2(s_2) \rightarrow A_3(s_3)\big)\\
                     &  & \big(\exists G\in \prod_{a_2\in A_2(s_2)} D_3(s_3,f(a_2)) \rightarrow D_2(s_2,a_2)\big)\\
                     &  & \big(\forall a_2\in A_2(s_2)\big) \ \big(\forall d_3\in D_3(s_3,f(a_2))\big)\\
                     &  & \big(\exists d_1\in D_1(s_1,d_1)\big)\\
                     &  & \quad \big(s_1[a_1/d_1],s_2[a_2/G_{a_2}(d_3)],s_3[f(a_2)/d_3]\big) \in r
  \ee
  which is equivalent to
  \be
  (s_1,s_2,s_3)\in r & \Rightarrow & \big(\forall a_1\in A_1(s_1)\big)\\
                     &  & \left(\exists (f,G) \in \begin{array}[c]{l}
                                                   \sum_{f\in A_2(s_2) \rightarrow A_3(s_3)}\\
                                                   \prod_{a_2\in A_2(s_2)}  D_3(s_3,f(a_2)) \rightarrow D_2(s_2,a_2)
                                                 \end{array}\right)\\
                     &  & \big(\forall (a_2,d_3) \in \sum_{A_2(s_2)} D_3(s_3,f(a_2))\big)\\
                     &  & \big(\exists d_1\in D_1(s_1,d_1)\big)\\
                     &  & \quad \big(s_1[a_1/d_1],s_2[a_2/G_{a_2}(d_3)],s_3[f(a_2)/d_3]\big) \in r
  \ee
  By definition, this means that $r$ is a simulation from $w_1$ to $w_2 \LinearArrow w_3$.

  \smallbreak
  \noindent
  Once more, all this formal manipulation keeps the computational content of
  the simulations. (Because $\AC$ is constructively valid.)
  \qed
\end{proof}
%
The notion of safety property from \cite{ISLL} corresponds to
simulations from $\One$ to $w$, or equivalently, subsets $x$ of $S$ such that:
\[
  s\in x  \Rightarrow  \big(\exists a\in A(s)\big)\big(\forall d\in D(s,a)\big)\ s[a/d]\in x
  \ \hbox{.}
\]
The analogy with strategies should be obvious: if $x$ is a
safety property, and $s\in x$ then the Angel has a strategy to avoid
deadlocks, starting from $s$.

\subsubsection{Multithreading.}

We now come to the last connective needed to interpret the $\lambda$-calculus.
Its computational interpretation is related to the notion of
\emph{multithreading,} \ie the possibility to run several instances of a
program in parallel. Let's start by defining synchronous multithreading in the
most obvious way:
\begin{definition}
  If $w$ is an interaction system on $S$, define $\LL(w)$, the multithreaded
  version of $w$ to be the interaction system on $\List(S)$ with:
  \be
    \LL.A\big((s_1,\ldots s_n)\big)                                           &=& A(s_1) \times \ldots A(s_n)\\
    \LL.D\big((s_1,\ldots s_n),(a_1,\ldots a_n)\big)                  &=& D(s_1,a_1) \times \ldots D(s_n,d_n)\\
    \LL.n\big((s_1,\ldots s_n),(a_1,\ldots a_n),(d_1,\ldots d_n)\big) &=& \big(s_1[a_1/d_1], \ldots s_n[a_n/d_n]\big) \ \hbox{.}
  \ee
\end{definition}
This interaction system is just an ``$n$-ary'' version of the synchronous
product. To get the abstract properties we want, we need to ``quotient''
multithreading by permutations. Just like multisets are list modulo
permutation, so is $!w$ the multithreaded $\LL(w)$ modulo permutations. This
definition is possible because $\LL(w)$ is ``compatible'' with permutations:
if $\sigma$ is a permutation, we have
\bce
 \sigma \cdot \big((s_1,\ldots s_n)\big[(a_1,\ldots a_n)/(d_1,\ldots d_n)\big]\big)\\
  =\\
  \big(\sigma \cdot (s_1,\ldots s_n)\big)[\sigma \cdot (a_1,\ldots a_n)/\sigma \cdot (d_1,\ldots d_n)]
  \ \hbox{.}
\ece
The final definition is:
\begin{definition}
  If $w$ is an interaction system on $S$, define $\LL(w)$, define $!w$ to be
  the following interaction system on $\Mulf(S)$:
  \be
  !A(\mu) &=& \sum_{\overline{s}\in\mu} \LL.A(\overline{s})\\
  !D\big(\mu,(\overline{s},\overline{a})\big) &=& \LL.D(\overline{s},\overline{a})\\
  !n\big(\mu,(\overline{s},\overline{a}),\overline{d}\big) &=& \Perm \cdot \LL.n(\overline{s},\overline{a},\overline{d}) \ \hbox{.}
  \ee
\end{definition}
Unfolded, it gives:
\begin{itemize}
  \item an action in state $\mu$ (a multiset) is given by an element $\overline{s}$ of
    the equivalence class $\mu$ (a list) together with an element $\overline{a}$ in
    $\LL.A(\overline{s})$ (a list of actions);
  \item a reaction is given by a list of reactions $\overline{d}$ in $\LL.D(\overline{s},\overline{a})$;
  \item the next state is the equivalence class containing the list
    $\overline{s}[\overline{a}/\overline{d}]$ (the orbit
    of~$\overline{s}[\overline{a}/\overline{d}]$ under the action of the group
    of permutations).
\end{itemize}

\noindent
This operation enjoys a very strong algebraic property:
\begin{proposition} \label{prop:OfCourse}
  ``\/$!\_$'' is a comonad in $\Int$.
\end{proposition}
\begin{proof}

  We need to find two operations:
  \begin{itemize}
   \item $\varepsilon_w : !w \rightarrow w$ defined as $\varepsilon_w = \big\{ \big([s],s\big) \mid s\in S\big\}$;
   \item and $\delta_w : !!w \rightarrow !w$ defined as the graph of the ``\texttt{concat}''
     function:
     \[\delta_w = \big\{ \big([\mu_i]_{i\in I} , \sum_{i\in I}\mu_i\big) \mid 
     \forall i\in I\  \mu_i\in\Mulf(S)\big\}\]
 \end{itemize}
 For any $w$, those operations are indeed simulations: for $\varepsilon_w$, it
 is quite obvious, and for $\delta_w$, it is quite painful to write. Let's only
 give an example from which the general case can easily be inferred:
 \begin{enumerate}
   \item we have $([ [s_1,s_2,s_3],[t_1],[] ] , [s_1,s_2,s_3,t_1]) \in \delta_w$
   \item for any command $((a_1,a_2,a_3),(b_1),())$ in state $[
     [s_1,s_2,s_3],[t_1],[] ]$, we need to find an action in
     $[s_1,s_2,s_3,t_1]$: simply take $(a_1,a_2,a_3,b_1)$;
   \item for any reaction $(d_1,d_2,d_3,e_1)$ to this action, we need to find
     a reaction to the original command, \ie to
     $((a_1,a_2,a_3),(b_1)())$: take $((d_1,d_2,d_3),(e_1),())$;
   \item the next states are respectively
     \begin{itemize}
       \item $[ [n(s_1,a_1,d_1),n(s_2,a_2,d_2),n(s_3,a_3,d_3)],[n(t_1,b_1,e_1)],[] ]$
       \item and $[n(s_1,a_1,d_1),n(s_2,a_2,d_2),n(s_3,a_3,d_3),n(t_1,b_1,e_1)]$.
     \end{itemize}
     They are indeed related through $\delta_w$.
 \end{enumerate}
 To be really precise, one would need to manipulate lists of states
 (representative of the multisets); but this only makes the proof even
 less readable.

 \smallbreak
 Checking that the appropriate diagrams commute is immediate. It only involves
 the underlying sets and relations, and not the interaction systems or
 simulation conditions. (In fact , finite multisets form a comonad in the
 category of sets and relations...)
\qed
\end{proof}

\section{Interpreting the $\lambda$-Calculus} 

We now have all the ingredients to give a denotational model for the typed
\mbox{$\lambda$-calculus}: a type $T$ will be interpreted by an interaction
system $T^*$; and a judgement ``$x_1:T_1,\dots x_n:T_n \vdash t:T$'' will be
interpreted by simulation from $!T_1^*\Tensor\dots !T_n^*$ to $T^*$.

\subsection{Typing rules}

The typing rules for the simply typed $\lambda$-calculus are given below:
\begin{enumerate}
  \item \infer{}{\Gamma \vdash x:\omega}{} if $x:\omega$ appears in $\Gamma$;

  \item \infer{\Gamma \vdash t:\omega \rightarrow \omega' & \Gamma \vdash u:\omega}{\Gamma \vdash (t)u:\omega'}{};

  \item \infer{\Gamma,x:\omega \vdash t:\omega'}{\Gamma \vdash \lambda x.t:\omega \rightarrow \omega'}{}.
\end{enumerate}
We follow Krivine's notation for the application and write ``$(t)u$'' for the
application of $t$ to $u$.

\subsection{Interpretation of Types}

We assume a set of type variables (``propositional variables''): $X, \ldots$
Nothing depend on the valuation we give to those type variables, so that we
are almost interpreting $\PI$ $\lambda$-calculus.\footnote{System-$F$ in which
all the quantifiers appear at the beginning of the term. To get an idea on how
to get a real model of system-$F$, refer to \cite{PTSecondOrder}.}

For a valuation $\rho$ from type variables to interaction systems, the
interpretation of types is defined in the usual way:
\begin{definition}
  Let $\omega$ be a type. Define the interpretation
  $\omega^*$ of $\omega$ as:
  \begin{itemize}
    \item $X^* = \rho(X)$;
    \item $(\omega \rightarrow \omega')^* = !\omega^* \LinearArrow \omega'^*$.
  \end{itemize}
\end{definition}

\subsection{Interpretation of Terms}

If $\omega$ is a type, write $|\omega|$ for the set of states of its interpretation:
\begin{itemize}
  \item $|X_i| = S_i$ (set of states of $\rho(X_i)$);
  \item $|\omega \rightarrow \omega'| = \big(\Mulf|\omega|\big) \times |\omega'|$.
\end{itemize}
A valuation is a way to interpret typed variables from the context:
\begin{definition}
  If $\Gamma=x_1:\omega_1,\ldots x_n:\omega_n$ is a context, an environment
  for~$\Gamma$ is a tuple $\gamma$ in $\Mulf|\omega_1| \times \ldots
  \Mulf|\omega_n|$. To simplify notation, we may write the tuple
  $\gamma=(\mu_1,\ldots \mu_n)$ as ``$x_1:=\mu_1,\ldots x_n:=\mu_n$''. We may
  also write $\gamma(x)$ for the projection of $\gamma$ on the appropriate
  coordinate. Sum of tuples of multisets is defined pointwise.
\end{definition}

We now interpret judgements: if we can type $\Gamma \vdash t:\omega'$ and if
$\gamma$ is an environment for $\Gamma$, the interpretation $[\![t]\!]_\gamma$ of
term $t$ in environment $\gamma$ is a subset of $|\omega|$ defined as follows:
\begin{definition}
  We define $[\![t]\!]_\gamma$ by induction on $t$:
  \begin{enumerate}
    \item if we have \infer{}{\Gamma \vdash x:\omega}{} with $x:\omega$ in
      $\Gamma$,\\
      then
      $[\![x]\!]_{\gamma} = \left\{\begin{array}[c]{lll}
                        \{s\}     & \mbox{if $\gamma(x)=[s]$ and $\gamma(y)=[]$ whenever $x\neq y$}\\
                        \emptyset & \mbox{otherwise} & \hbox{;}
                      \end{array}\right.$

    \item if we have \infer{\Gamma \vdash t:\omega \rightarrow \omega' & \Gamma \vdash u:\omega}{\Gamma \vdash 
      (t)u:\omega'}{},\\
      then
      $s\in[\![(t)u]\!]_\gamma$ iff $(\mu,s)\in[\![t]\!]_{\gamma_0}$ for some
      $\mu=[s_1,\ldots s_n] \in \Mulf|\omega|$ s.t.  $s_i\in[\![u]\!]_{\gamma_i}$
      for all $i=1,\ldots n$ and $\gamma=\gamma_0+\gamma_1+\ldots\gamma_n$;


    \item if we have \infer{\Gamma,x:\omega \vdash t:\omega'}{\Gamma \vdash \lambda x.t:\omega \rightarrow \omega'}{},\\
      then
      $[\![\lambda x.t]\!]_\gamma = \{ (\mu,s) \mid \mu\in\Mulf|\omega|,\,
      s\in[\![t]\!]_{\gamma,x:=\mu}\}$.
  \end{enumerate}
\end{definition}
It is immediate to check that this definition is well formed.

\medbreak
If $\Gamma =x_1:\omega_1,\ldots x_n:\omega_n$, write $!\Gamma$ for
$!\omega_1^* \Tensor \ldots !\omega_n$; similarly, we omit the superscript
$\_^*$ and write $\omega$ for $\omega^*$. The interpretation of terms is
correct in the following sense:

\begin{proposition} \label{prop:correct}
  Suppose that $\Gamma \vdash t:\omega'$, then the relation ``$\_ \in [\![t]\!]_{\_}$''
  is a simulation relation from $!\Gamma$ to $\omega'$.
\smallbreak \noindent
  In other words, if $s\in[\![t]\!]_{\gamma}$, then $s$ (in $\omega'$) simulates
  $\gamma$ (in $!\Gamma$).
\end{proposition}
This is quite surprising because the interpretation of $t$ doesn't depend on
the interaction systems used to interpret the types but only the underlying
set of states.\footnote{The interpretation is called the relational
interpretation: it can be defined in the category of sets and relations...}

\begin{proof}
  We work by induction on the structure of the type inference.
  \begin{enumerate}

    \item Axiom: it amount to showing that $\{([],\ldots [],[s],[],\ldots [],
      s) \mid s\in|\omega|\}$ is a simulation from $!\Gamma$ to $\omega$.
      This is easy: the only actions available in state
      $([],\ldots[s],[],\ldots)$ are of the form $((),\ldots(a),()\ldots)$
      where $a\in A(s)$, and they are simulated by the action $a$.  The
      reaction $d$ is translated back into reaction
      $((),\ldots,(d),(),\ldots)$; and the rest is obvious.

\medbreak

    \item Application: suppose we have $s\in [\![(t)u]\!]_\gamma$. By definition,
      we know that we have $(\mu,s)\in[\![t]\!]_{\gamma_0}$ for some
      $\mu=[s_1,\ldots s_n]$ s.t. each $s_i$ is in $[\![u]\!]_{\gamma_i}$ for a
      partition $\gamma = \gamma_0+\gamma_1+\ldots \gamma_n$.

      By induction hypothesis, we thus know that $(\mu,s)$ (in
      $\omega \rightarrow \omega'$) simulates $\gamma_0$ (in $!\Gamma$); and that any
      $s_i$ (in $\omega$) simulates $\gamma_i$ (in $!\Gamma$).

      \smallbreak
      Rather than doing the full formal proof (which involves many indices),
      we'll show how it works on an example. The general case can easily be
      deduced from that.

      Suppose $\Gamma$ is reduced to a single assumption $x:\nu$ so that $\gamma$
      is reduced to a single multiset, $[v_1,v_2,v_3]$ for our example.
      Suppose $s\in[\![(t)u]\!]_{x:=[v_1,v_2,v_3]}$ because:
      \begin{itemize}
        \item $([t_1,t_2],s)\in[\![t]\!]_{x:=[v_2]}$
        \item $t_1\in[\![u]\!]_{x:=[v_1,v_3]}$ and
          $t_2\in[\![u]\!]_{x:=[]}$.
      \end{itemize}
      We need to show that $s$ simulates $[v_1,v_2,v_3]$:
      \begin{enumerate}
        \item suppose $a_1\in A_\nu(v_1)$, $a_2\in A_\nu(v_2)$ and $a_3\in
          A_\nu(v_3)$;
        \item we need to find an action in $A_{\omega'}(s)$ simulating
          $(a_1,a_2,a_3)$:
          \begin{itemize}
             \sitem1 by induction hypothesis, $t_1$ simulates $[v_1,v_3]$,
               so that we can find an action $b_1\in A_{\omega}(t_1)$ simulating
               $(a_1,a_3)$;

             \sitem2 similarly, $t_2$ simulates $[]$, so that we can find
               an action $b_2\in A_{\omega}(t_2)$ simulating $()$;

             \sitem3 we also have that $([t_1,t_2],s)$ (in
             $\omega \rightarrow \omega'$) simulates $[v_2]$ (in $!\nu$). By
             proposition~\ref{prop:adjoint}, this is equivalent to saying that
             $s$ (in $\omega'$) simulates $([v_2],[t_1,t_2])$ (in
             $!\nu\Tensor !\omega$).

             Thus, we can find an action $a\in A_{\omega'}(s)$ simulating
             $\big((a_2),(b_1,b_2)\big)$.

             By composing the above two simulations on the right ($(b_1,b_2)$
             simulates $(a_1,a_3)$), we thus obtain that $a$ simulates
             $(a_1,a_2,a_3)$.

          \end{itemize}
          We now need to translate the reactions back: let $d\in D_{\omega'}(s,a)$,

          \begin{itemize}
            \sitem3 by induction, we can translate $d$ into a reaction
            $\big((d_2),(e_1,e_2)\big)$ to $\big((a_2),(b_1,b_2)\big)$;

            \sitem2 we can translate $e_2$ into a reaction $()$ to $b_2$;

            \sitem1 and finally we can translate $e_1$ into a reaction $(d_1,d_3)$ to
            $(a_1,a_3)$.

          \end{itemize}
          Thus, we obtain reactions $d_1\in D_\nu(v_1,a_1)$, $d_2\in
          D_\nu(v_2,a_2)$ and
          $d_3\in D_\nu(v_3,a_3)$.

        \item The new states we get from those actions/reactions are: $s[a/d]$
          on one side; and $[v_1[a_1/d_1],v_2[a_2/d_2],v_3[a_3/d_3] ]$ on the
          other side.  They are indeed related because:
          \begin{itemize}
            \sitem1 $t_1[b_1/e_1] \in [\![u]\!]_{x:=[v_1[a_1/d_1],v_3[a_3/d_3] ]}$;

            \sitem2 $t_2[b_2/e_2] \in [\![u]\!]_{x:=[]}$;

            \sitem3 and finally $[t_1[b_1/e_1],t_2[b_2/e_2] ] \in [\![t]\!]_{x:=v_2[a_2/d_2]}$.
           \end{itemize}
      \end{enumerate}

\medbreak

    \item Abstraction: this is immediate. Suppose $(\mu,s)\in[\![\lambda
      x.t]\!]_{\gamma}$; we need to show that $(\mu,s)$ (in
      $\omega \rightarrow \omega'$) simulates $\gamma$ (in $!\Gamma$). By proposition
      \ref{prop:adjoint}, this is equivalent to showing that $s$ (in
      $\omega'$) simulates $(\gamma,\mu)$ (in $!\Gamma\Tensor!w$). This is
      exactly the induction hypothesis.

    \qed
  \end{enumerate}
\end{proof}
To summarise all this, here is a tentative rewording of the above: if $\Gamma
 \vdash t:\omega$,
\smallbreak
\bgroup
  \sl
  \sit1 each type represent a process;

  \sit2 each process in the context can be run in parallel multiple times;

  \sit3 the environment $\gamma$ represents the initial states for the
  context;

  \sit4 if $s \in [\![t]\!]_\gamma$ then $s$ can be used as an initial state to
  simulate $\gamma$;

  \sit5 the algorithm for the simulation is contained in $t$.
\egroup

\bigbreak
To finish the justification that we have a denotational model, we now need
to check that the interpretation is invariant by $\beta$-reduction.
\begin{proposition} \label{prop:red_inv}
  For all terms $t$ and $u$ and environment $\gamma$, we have
  \[
    [\![(\lambda x.t)u]\!]_\gamma = [\![ t[u/x] ]\!]_\gamma
    \ \hbox{.}
  \]
\end{proposition}
The proof works by induction and is neither really difficult nor very
interesting. It can be found on
\texttt{http://iml.univ-mrs.fr/$\sim$hyvernat/academics.html}.

\section{Interpreting the Differential $\lambda$-calculus} 

Simulation relations from $w$ to $w'$ enjoy the additional property that they
form a complete sup-lattice:
\begin{lemma} \label{lem:supLattice}
  The empty relation is always a simulation from any $w$ to $w'$; and if
  $(r_i)_{i\in I}$ is a family of simulations from $w$ to $w'$, then
  $\bigcup_{i\in I} r_i$ is also a simulation from $w$ to $w'$.
\end{lemma}
The proof is immediate...

\medbreak
Unfortunately, this doesn't reflect any property of $\lambda$-terms. The
reason is that \sit1 not every type is inhabited, and \sit2 we do not see
\textit{a priori} how to take the union of two terms. For example, what is the
meaning of $\lambda x\lambda y.x \cup \lambda x\lambda y.y$ in the type
$X \rightarrow X \rightarrow X$?\footnote{In terms of usual datatypes translation, this term would
be $\T\cup\F$ in the type $\Bool$.}

Ehrhard and Regnier's \emph{differential $\lambda$-calculus} (\cite{diffLamb})
extends the $\lambda$-calculus by adding a notion of differentiation of
$\lambda$-terms. One consequence is that we need to have a notion of sum of
arbitrary terms, interpreted as a non-deterministic choice. It is not the right
place to go into the details of the differential $\lambda$-calculus and we
refer to \cite{diffLamb} for motivations and a complete description.

\smallbreak
\noindent
In the typed case, we have the following typing rules:
\begin{enumerate}
  \item \infer{}{\Gamma \vdash 0:\omega}{} and \infer{\Gamma \vdash t:\omega & \Gamma \vdash u:\omega}{\Gamma \vdash t+u:\omega}{};

  \item \infer{\Gamma \vdash t:\omega \rightarrow \omega' & \Gamma \vdash u:\omega}{\Gamma \vdash \D t \cdot u : \omega \rightarrow \omega'}{}.
\end{enumerate}
The intuitive meaning is that ``$\D t \cdot u$'' is the result of
(non-deterministically) replacing \emph{exactly one occurrence} of the first
variable of $t$ by $u$. We thus obtain a sum of terms, depending on which
occurrence was replaced.
This gives a notion of differential substitution (or linear substitution)
which yields a \emph{differential-reduction.} The rules governing this
reduction are more complex than usual \mbox{$\beta$-reduction} rules.  We refer to
\cite{diffLamb} for a detailed description.

\smallbreak
\noindent
We extend the interpretation of terms in the following way:
\begin{definition}
  Define the interpretation of a typed differential $\lambda$-term by
  induction on the type inference:
  \begin{enumerate}

    \item if we have \infer{}{\Gamma \vdash 0:\omega}{},
      then we put $[\![0]\!]_\gamma = \emptyset$;

    \item if we have \infer{\Gamma \vdash t:\omega & \Gamma \vdash u:\omega}{\Gamma \vdash 
      t+u:\omega}{},\\
      then we put $[\![t+u]\!]_\gamma = [\![t]\!]_\gamma \cup [\![u]\!]_\gamma$;

    \item if we have \infer{\Gamma \vdash t:\omega \rightarrow \omega' & \Gamma \vdash u:\omega}{\Gamma \vdash \D
      t \cdot u : \omega \rightarrow \omega'}{},\\
      then we put
      $(\mu,s') \in [\![\D t \cdot u]\!]_{\gamma}$ iff $(\mu+[s],s') \in
      [\![t]\!]_{\gamma_1}$ for some $s\in[\![u]\!]_{\gamma_2}$ s.t.
      $\gamma=\gamma_1+\gamma_2$.
  \end{enumerate}
\end{definition}

\noindent
Proposition \ref{prop:correct} extends as well:
\begin{proposition}
  Suppose that $\Gamma \vdash t:\omega'$ where $\Gamma$ is a context and $t$ a
  differential $\lambda$-term. The relation ``$\_ \in [\![t]\!]_{\_}$'' is a
  simulation relation from $!\Gamma$ to $\omega'$.
\end{proposition}

\begin{proof}
  The proof for the sum and the $0$ are contained in proposition
  \ref{lem:supLattice}.

  \smallbreak\noindent
  For differentiation, suppose we have $(\mu,s')\in[\![\D t \cdot u]\!]_\gamma$, \ie
  $(\mu+[s],s')\in [\![t]\!]_{\gamma_1}$ for some $s\in[\![u]\!]_{\gamma_2}$, with
  $\gamma=\gamma_1+\gamma_2$.
  We need to show that $(\mu,s')$ (in $\omega \rightarrow \omega'$) simulates $\gamma$ (in
  $!\Gamma$). Since $\gamma=\gamma_1+\gamma_2$, it is enough to show that we
  can simulate $(\gamma_1,\gamma_2)$ (in $!\Gamma\Tensor!\Gamma$).
  \\
  By proposition~\ref{prop:adjoint}, this is equivalent to showing that $s'$
  (in~$\omega'$) simulates $(\gamma_1,\gamma_2,\mu)$
  (in~$!\Gamma\Tensor!\Gamma\Tensor!\omega$).
  \\
  Let $a_{\gamma_1}\in !A_\Gamma(\gamma_1)$, $a_{\gamma_2}\in
  !A_\Gamma(\gamma_2)$ and $a_\mu\in!A_\omega(\mu)$; we need to find an action
  in $A_{\omega'}(s')$ to simulate $(a_{\gamma_1},a_{\gamma_2},a_\mu)$:
  \begin{itemize}

    \sitem1 by induction hypothesis, we know that $s$ (in $\omega$) simulates
    $\gamma_2$ (in $!\Gamma$); so that we can find an action $a\in
    A_\omega(s)$ simulating $a_{\gamma_2}$;

    \sitem2 by induction, we know that $s'$ (in $\omega'$)
    simulates $(\gamma_1,\mu+[s])$ (in $!\Gamma\Tensor!\omega$), so that we
    can find an action $a'\in A_{\omega'}(s')$ simulating
    $\big(a_{\gamma_1},(a_\mu,a)\big)$.

    Since $a$ simulates $a_{\gamma_2}$, by composition, $a$ simulates
    $\big(a_{\gamma_1},(a_\mu,a_{\gamma_2})\big)$; and by associativity and
    commutativity, we can thus simulate $(a_{\gamma_1},a_{\gamma_2},a_\mu)$.

  \end{itemize}
  To translate back a reaction $d'$ to $a'$ into a reaction
  $(d_{\gamma_1},d_{\gamma_2},d_\mu)$, we proceed similarly:
  \begin{itemize}

    \sitem2 by induction, we can translate $d'$ into a reaction
    $(d_{\gamma_1},d_\mu,d)$ to $\big(a_{\gamma_1},(a_\mu,a)\big)$;

    \sitem1 by induction, we can also translate the reaction $d$ (in
    $D_\omega(s,a)$) into a reaction $d_{\gamma_2}$ (in
    $!D_{\Gamma}(s,a_{\gamma_2})$).

  \end{itemize}
  We thus obtain reactions $d_{\gamma_1}$, $d_{\gamma_2}$ and $d_\mu$ as
  desired. That the resulting next states are still related is quite
  obvious...  \qed
\end{proof}

We now need to check that the interpretation is invariant by $\beta$-reduction
and differential reduction.
\begin{proposition}
  For all differential terms $t$ and $u$ and environment $\gamma$, we have:
  \be
    [\![ (\lambda x.t)u ]\!]_\gamma    &=& [\![ t[u/x] ]\!]_\gamma \cr
    [\![ \D(\lambda x.t) \cdot u ]\!] _\gamma &=& [\![ \lambda x \ .\  (\partial t/\partial x) \cdot u ]\!]_\gamma
  \ee
\end{proposition}
Just like for Proposition~\ref{prop:red_inv}, the proof is quite easy but
tedious. The interested reader can find it at
\texttt{http://iml.univ-mrs.fr/$\sim$hyvernat/academics.html}.


\section*{Conclusion} 

Technically speaking, this work is not very different from \cite{ISLL}, which
is itself quite close to \cite{denotPT}. The main reasons for producing it
are:
\begin{itemize}
  \item first, it shows that we can give a computational content to the notion
    of simulation if we do not try to interpret all of linear logic;

  \item second, it shows that some of the additional structure of interaction
    systems and simulation does have a logical significance. We showed that by
    interpreting the differential $\lambda$-calculus.
\end{itemize}
Even if we haven't done it formally, it is quite easy to extend the model to
full intuitionistic linear logic while keeping the computational content
of simulations. To define the additive, we use the definition of $\Plus$ from
\cite{ISLL}.

\smallbreak
It is in principle possible to formalise all the above in a proof assistant
(Agda \cite{Agda} or Coq \cite{Coq} come to mind).\footnote{One needs to be
careful to be able to deal with the notion of equivalence classes used in the
definition of $!w$. The idea is to use interaction systems on ``setoids'',
where the equivalence relation is a simulation...} From such a system, one
could extract the simulations. For example, a term of type $T \rightarrow T'$ would
give an algorithm simulating many synchronous occurrences of $T$ by a single
occurrence of $T'$.

It is however difficult to apply this to obtain real-life simulations. The
problem is that we only get ``purely logical'' simulations. Simulations of
interest for application rely heavily on the different interaction systems
used. One way to get more interesting simulations (from a practical point of
view) might be to use constant interaction systems (booleans, natural numbers,
or more practical ones like stacks, memory cells, \etc.) as ground types,
together with specific simulations (the values true and false, successor
function, or more practical simulations) as inhabitant of specific types.

\medbreak
In pretty much the same way as \cite{ISLL} makes $\Int$ into a denotational
model for classical linear logic, we can make interaction systems into a
denotational model for ``classical differential linear logic'': differential
interaction nets \cite{diffNets}. This system doesn't make much sense
logically speaking, but seems to enjoy relationship with process calculi.
This is an encouraging direction of research.




\end{document}